\documentclass[runningheads]{llncs}
\usepackage{graphicx}
\usepackage{amsmath}
\usepackage{amssymb}

\usepackage[ruled,vlined,linesnumbered,commentsnumbered]{algorithm2e}

\usepackage[capitalise,nameinlink]{cleveref}

\usepackage{graphicx}
\usepackage{color}
\usepackage[dvipsnames]{xcolor}

\usepackage[]{todonotes}

\usepackage{tikz}
\usetikzlibrary{calc}
\usetikzlibrary{arrows,shapes}
\usetikzlibrary{decorations.markings}
\usetikzlibrary{arrows.meta}
\usetikzlibrary{shapes.misc}

\usepackage{subcaption}
\begin{document}

\newcommand{\X}{{\mathcal{X}}}
\newcommand{\cU}{{\mathcal{U}}}
\newcommand{\cI}{{\mathcal{I}}}
\newcommand{\cC}{{\mathcal{C}}}
\newcommand{\rev}[1]{\textcolor{blue}{#1}} 
\newcommand{\R}{\mathbb{R}}
\newcommand{\N}{\mathbb{N}}

\newcommand{\probl}[1]{\textsc{#1}}
\newcommand{\Pint}{\operatorname{int}}
\newcommand{\Pends}{\operatorname{ends}}
\newcommand{\opt}{\text{OPT}}

\title{On the Recoverable Traveling Salesman Problem}

\author{Marc Goerigk\inst{1}\orcidID{0000-0002-2516-0129} \and
Stefan Lendl\inst{2}\orcidID{0000-0002-5660-5397} \and
Lasse Wulf\inst{3}\orcidID{0000-0001-7139-4092}}

\institute{Network and Data Science Management, University of Siegen, Unteres Schlo{\ss}~3, 57072 Siegen, Germany 
\email{marc.goerigk@uni-siegen.de} \and
Institute of Operations and Information Systems, University of Graz, Universit\"atsstra{\ss}e~15, 8010 Graz, Austria
\email{stefan.lendl@uni-graz.at} \and
Institute of Discrete Mathematics, Graz University of Technology, Steyrergasse~30/II, 8010 Graz, Austria \email{wulf@math.tugraz.at}}

\maketitle

\begin{abstract}
In this paper we consider the \probl{Recoverable Traveling Salesman Problem} (TSP). Here the task is to find two tours simultaneously, 
such that the intersection between the tours is at least a given minimum size, 
while the sum of travel distances with respect to two different distance metrics is minimized. 
Building upon the classic double-tree method, we derive a 4-approximation algorithm for the \probl{Recoverable TSP}. We also show that if the required size of the intersection between the tours is constant, a 2-approximation guarantee can be achieved, even if more than two tours need to be constructed. We discuss consequences for approximability results in the more general area of recoverable robust optimization.
\end{abstract}

\noindent\textbf{Keywords:} recoverable robustness; intersection constraints; traveling salesman problem; approximation


\section{Introduction}

Uncertainty and incomplete problem knowledge can have significant impact on the quality of decision we make. Several paradigms have been developed to include data uncertainty in the decision making process, including recoverable robustness. In this setting, we construct a first solution while data uncertainty is still present, and can later adjust this solution in a second stage, when full problem knowledge is available.

In principle, this approach can be applied to any (combinatorial) optimization problem, including the classic \probl{Traveling Salesman Problem (TSP)}. The \probl{TSP} is a well-studied fundamental problem in  combinatorial optimization, computer science and operations research. We denote by $(V, d)$ a \probl{TSP} instance, if  
$V$ is a set of vertices (cities) and $d \colon \binom{V}{2} \rightarrow \R_{\geq 0}$ a distance function.
We call $C = (v_0, v_1, \dots, v_n )$ a tour on $V$, if $|V|=n$, $v_0 = v_n$ and $V = \{v_0, \dots, v_{n-1} \}$.
We denote by $E(C) := \{ \{v_i, v_{i+1}\} \colon i=0,\dots,n-1 \}$ the edge set of $C$.
The \probl{TSP} asks for a tour $C$ of the vertices minimizing 
$d(C) = \sum_{e \in E(C)} d(e)$. The \probl{TSP} is known to be NP-hard and is 
one of the most-studied problems with respect to approximation algorithms. In its general 
form it is inapproximable, but if the distance function $d$ is a metric on $V$, then 
different constant factor approximation algorithms are known. A $2$-approximation can be achieved 
by the classic double-tree shortcutting algorithm~\cite{rosenkrantz1977analysis} and a 
$3/2$-approximation was achieved in the seminal work~\cite{christofides1976worst} 
introducing Christofides' algorithm. Up until 2020 this was the best-known approximation guarantee 
for the general metric \probl{TSP}, when a $3/2-\varepsilon$-approximation was achieved~\cite{karlin2021slightly}.

In this paper, we study a variant of \probl{TSP} denoted as 
\probl{Recoverable Traveling Salesman Problem (RecovTSP)}, which was introduced in 
the area of recoverable robust optimization~\cite{chassein2016recoverable}.
We denote by $(V,d_1,d_2,q)$ a \probl{RecovTSP} instance, if $V$ is a set of vertices (cities),
$d_1,d_2 \colon \binom{V}{2} \rightarrow \R_{\geq 0}$ are distance functions and $q \in \N$ is 
the intersection size parameter. The \probl{RecovTSP} asks for two tours $C_1, C_2$ on $V$ 
that minimize the objective $d_1(C_1) + d_2(C_2)$, subject to the constraint that $C_1$ and $C_2$ 
have at least $q$ edges in common, i.e. $|E(C_1) \cap E(C_2)| \geq q$. If both $d_1$ and $d_2$ are metric distance 
functions the given instance is an instance of \probl{Metric RecovTSP}.

We also study a multi-stage generalization of \probl{RecovTSP} which we call
\probl{$k$-Stage Recoverable Traveling Salesman Problem ($k$-St-RecovTSP)}. Here, $k \in \N$ is 
a given number of stages and an instance $(V, d_1, \dots, d_k, q)$ of \probl{$k$-St-RecovTSP} 
asks for $k$ tours $C_1, \dots, C_k$ that minimize the objective $\sum_{j=1}^{k} d_j(C_j)$,
subject to the constraint $\left|\bigcap_{1\leq j \leq k} E(C_j) \right| \geq q$.


\paragraph{Related Results.} The study of discrete optimization problems with intersection constraints, known 
as recoverable optimization problems, was initiated through their application in 
recoverable robust optimization~\cite{liebchen2009concept}. To define a recoverable robust problem, it is necessary to specify the set of scenarios that we wish to consider and protect against. The complexity of the resulting robust problem then depends on the choice of this uncertainty set. We refer to the survey \cite{kasperski2016robust} for an overview on existing complexity results, where the recoverable robust problem is denoted as robust optimization with incremental recourse. The \probl{RecovTSP} we consider is equivalent to a recoverable robust problem with a single scenario or with interval uncertainty, an observation that will be explained more formally in \cref{sec:robust}.

In \cite{kasperski2017robust}, the \probl{Recoverable Robust Selection Problem} was considered under discrete and interval uncertainty. In the \probl{Selection Problem}, the set of feasible solutions consists of all choices of exactly $p$ out of $n$ items. They show that in the case of discrete uncertainty, the problem becomes strongly NP-hard, while it can be solved in $O(qn^2)$ for interval uncertainty, where $q$ is the size of the intersection between the two solutions. Recently, \cite{lachmann2021linear} further improved this solution time to $O(n)$.
The setting of recoverable robustness with interval uncertainty has also been considered in the context of the spanning tree problem. In the \probl{Recoverable Spanning Tree Problem} (\probl{RecovST}), one is given a vertex set $V$, two distance functions $d_1, d_2$ on $V$ and an intersection size parameter $q \in \N$. The goal is to find two spanning
trees $T_1, T_2$ on $V$ such that $|T_1 \cap T_2| \geq q$ and $d_1(T_1) + d_2(T_2)$ 
is minimized. In \cite{hradovich2017recoverable}, it was proven that \probl{RecovST} can be solved optimally in polynomial time $O(qm^2n)$, where $m$ is the number of edges.

The recoverable robust setting was also considered for the \probl{Assignment Problem} in \cite{fischer2020investigation}, where the authors show W[1]-hardness and present special cases that can be solved in polynomial time. In a related single-machine \probl{Scheduling Problem} setting, \cite{bold2021recoverable} derive a 2-approximation algorithm. The \probl{Recoverable Robust Shortest Path Problem} was studied in \cite{busing2012recoverable}, where is was shown that the problem becomes NP-hard and not approximable in most settings.
The recoverable setting can also be considered for matroid bases \cite{busing2011phd}.
It has been generalized to other intersection constraints in the context of matroid bases \cite{lendl2021matroid}, where a strongly polynomial algorithm is presented for the case of a lower bound on the intersection. This setting was further generalized in \cite{iwamasa2021optimal}, where also nonlinear convex cost functions were considered.

So far, only little attention has been given to the \probl{RecovTSP} -- possibly, as the underlying problem is already hard for the non-robust setting.  In the short paper \cite{chassein2016recoverable}, different solution methods were proposed for a recoverable robust setting with so-called budgeted uncertainty sets, which generalize interval uncertainty. To the best of our knowledge, no previous work has derived approximation results for this setting.

\paragraph{Our Contributions.}

This paper is the first to provide complexity results for the \probl{RecovTSP}. If the required size $q$ of the intersection between the two solutions is part of the input, we show that there exists a polynomial time 4-approximation algorithm. We provide an example that shows that our analysis of the algorithm is tight. If $q$ is a constant number, an improved 2-approximation algorithm can be achieved, based on enumerating all possible intersection sets. This algorithm also extends to a more general setting, where an arbitrary number of tours need to be found, instead of only two. Finally, we discuss consequences of our results in the area of robust optimization.

\section{A 4-Approximation Algorithm for the \probl{RecovTSP}}
\label{sec:recovtspapprox}

\subsection{Main Result and Proof Idea}

In this section we prove the following result.

\begin{theorem}\label{thm:tsp:approx}
    There is a $4$-approximation algorithm for \probl{Metric RecovTSP}.
\end{theorem}

We first explain the idea of the algorithm before we describe it formally. The starting point of our algorithm is the \probl{RecovST}. Given an instance $(V, d_1, d_2, q)$ of the \probl{Recoverable TSP}, we begin by obtaining an optimal solution $(T_1, T_2)$ to the \probl{RecovST} with the same parameters $(V, d_1, d_2, q)$. The trees $T_1$ and $T_2$ already have a sufficiently large intersection $|T_1 \cap T_2| \geq q$, so we would like to transform them into tours $C_1, C_2$ with the same intersection. However, there is a problem: The vertices in $T_1 \cap T_2$ could have degree greater than 2 in $T_1 \cap T_2$. For example, every component of $T_1 \cap T_2$ could be a star. But clearly in the final tours $C_1$ and $C_2$ 
every vertex must have degree 2 (here it is important to note that in the \probl{RecovTSP}, we do not allow a TSP tour to travel along the same edge multiple times). 
So $T_1 \cap T_2 = E(C_1) \cap E(C_2)$ is not possible in general. Usually, in the TSP literature, a tree is transformed into a tour by making it Eulerian (either considering the double-tree or inserting matching edges in Christofides' algorithm) and shortcutting an Eulerian tour of the resulting graph. However, observe that the shortcutting procedure will result in different outcomes for the trees $T_1$ and $T_2$, because they are different outside of their common intersection $T_1 \cap T_2$. Therefore, it is easy to see that the procedure of shortcutting does not preserve the property that $T_1$ and $T_2$ have a sufficiently large intersection. Let $K_1, \dots, K_r$ be the connected components of $T_1 \cap T_2$. As described above, the problem is that $K_j$ is not necessarily a path for $j = 1, \dots, r$. We can solve this problem by actually replacing $K_j$ with some Hamilton path $P_j$ that traverses all vertices of $K_j$. In order to find $P_j$, we can approximately solve the TSP problem on the sub-instance on vertex set $K_j$ and metric $d_1 + d_2$. 
Replacing each component $K_j$ by $P_j$, we obtain some new graph from $T_i$, which we call $T_i'$. We will show that if we double all the edges of $T_i'$, we can obtain an Eulerian circuit $W''_i$, which contains each of the paths $P_1, \dots, P_r$ as simple subpaths. We finally show that it is possible to shortcut $W''_i$ into a tour $C_i$ in such a way that the subpaths $P_1, \dots, P_r$ are preserved. Then we have that $P_1, \dots, P_r \subseteq C_i$ for $i = 1,2$, and therefore $|E(C_1) \cap E(C_2)| \geq q$. 
We will finally show that during the whole procedure we lose at most a constant factor compared to the optimal solution of \probl{RecovTSP}. Algorithm~\ref{algo:tsp:approx} provides a description in pseudo-code and \cref{fig:algorithm} depicts an example run of the algorithm.
\tikzstyle{vertex}=[draw,circle,fill=black, minimum size=4pt,inner sep=0pt]
\tikzstyle{edge} = [draw,-]
\tikzstyle{T1edge} = [draw,blue!60,densely dashed,line width=2,-]
\tikzstyle{T2edge} = [draw,OliveGreen,loosely dotted, line width=2,-]
\tikzstyle{IntersectionEdge} = [draw,red!50,line width=2.5,-]
\tikzset{IntersectionEdgeDirected/.style={->,red!50,red!50,line width=2.5,> = stealth}}
\tikzset{IntersectionEdgeDirectedThick/.style={->,red!50,red!50,line width=3,> = stealth}}
\tikzset{T1EdgeDirected/.style={->,draw,blue!60,dashed,line width=2,> = stealth}}
\tikzstyle{weight} = [font=\small]
\begin{figure}[tbhp]
\begin{subfigure}{0.45\textwidth}
\centering
\begin{tikzpicture}[scale=0.43, auto,swap]

    \foreach \pos/\name in {{(3,8)/v1}, {(11,8)/v2}, {(1,3)/v3}, {(3,6)/v4}, 
                            {(5,5)/v5}, {(7,6)/v6}, {(9,5)/v7}, {(11,6)/v8}, 									{(5,3)/v9}, {(9,3)/v10}, {(13,3)/v11}, {(3,2)/v12},
                            {(7,2)/v13}, {(11,2)/v14}}{
        \node[vertex] (\name) at \pos {};
    }
    \foreach \source/\dest in {v1/v4, v5/v6, v6/v7, v9/v13, v2/v8}{
        \draw[T1edge] (\source) to (\dest);
    }
    \foreach \source/\dest in {v5/v10, v13/v10, v5/v1, v1/v6, v2/v7}{
        \draw[T2edge] (\source) to (\dest);
    }
    \foreach \source/ \dest in {v3/v5, v4/v5, v5/v9, v5/v12, v7/v8, v8/v14, 
   								v10/v14, v14/v11}{
       \draw[IntersectionEdge] (\source) to (\dest);
    }
\end{tikzpicture}
 \subcaption{The tree $T_1$ (blue dashed edges and red edges), the tree $T_2$ (green dotted edges and red edges) and their intersection $T_1 \cap T_2$ (red edges).}
 \label{subfig:a}
\end{subfigure} \hfill
\begin{subfigure}{0.45\textwidth}
\centering
\begin{tikzpicture}[scale=0.43, auto,swap]

    \foreach \pos/\name in {{(3,8)/v1}, {(11,8)/v2}, {(1,3)/v3}, {(3,6)/v4}, 
                            {(5,5)/v5}, {(7,6)/v6}, {(9,5)/v7}, {(11,6)/v8}, 									{(5,3)/v9}, {(9,3)/v10}, {(13,3)/v11}, {(3,2)/v12},
                            {(7,2)/v13}, {(11,2)/v14}}{
        \node[vertex] (\name) at \pos {};
    }
    \foreach \source/\dest in {v1/v4, v5/v6, v6/v7, v9/v13, v2/v8}{
        \draw[T1edge] (\source) to (\dest);
    }
    \foreach \source/ \dest in {v3/v12, v4/v5, v5/v9, v9/v12, v7/v8, v7/v10, 
   								v10/v14, v14/v11}{
       \draw[IntersectionEdge] (\source) to (\dest);
    }
\end{tikzpicture}
 \label{subfig:b}
 \subcaption{The tree $T_1'$ is created from $T_1$ by substituting the connected components of $T_1 \cap T_2$ with simple paths $P_1, \dots, P_r$.}
\end{subfigure}
\par\bigskip
\begin{subfigure}{0.45\textwidth}
\centering
\begin{tikzpicture}[scale=0.43, auto,swap]

    \foreach \pos/\name in {{(3,8)/v1}, {(11,8)/v2}, {(1,3)/v3}, {(3,6)/v4}, 
                            {(5,5)/v5}, {(7,6)/v6}, {(9,5)/v7}, {(11,6)/v8}, 									{(5,3)/v9}, {(9,3)/v10}, {(13,3)/v11}, {(3,2)/v12},
                            {(7,2)/v13}, {(11,2)/v14}}{
        \node[vertex] (\name) at \pos {};
    }
    \foreach \source/\dest in {v1/v4, v5/v6, v6/v7, v9/v13, v2/v8}{
        \draw[T1edge, bend left = 20] (\source) to (\dest);
        \draw[T1edge, bend left = 20] (\dest) to (\source);
    }
    \foreach \source/ \dest in {v3/v12, v4/v5, v5/v9, v9/v12, v7/v8, v7/v10, 
   								v10/v14, v14/v11}{
       \draw[IntersectionEdge, bend left = 20] (\source) to (\dest);
       \draw[IntersectionEdge, bend left = 20] (\dest) to (\source);
    }
\end{tikzpicture}
 \label{subfig:c}
 \subcaption{The graph $T_1''$ is created from $T_1'$ by doubling all its edges.}
\end{subfigure} \hfill
\begin{subfigure}{0.45\textwidth}
\centering
\begin{tikzpicture}[scale=0.43, auto,swap]

    \foreach \pos/\name in {{(3,8)/v1}, {(11,8)/v2}, {(1,3)/v3}, {(3,6)/v4}, 
                            {(5,5)/v5}, {(7,6)/v6}, {(9,5)/v7}, {(11,6)/v8}, 									{(5,3)/v9}, {(9,3)/v10}, {(13,3)/v11}, {(3,2)/v12},
                            {(7,2)/v13}, {(11,2)/v14}}{
        \node[vertex] (\name) at \pos {};
    }
    \foreach \source/\dest in {v1/v4, v5/v6, v6/v7, v9/v13, v2/v8}{
        \draw[T1EdgeDirected, bend left = 20] (\source) to (\dest);
        \draw[T1EdgeDirected, bend left = 20] (\dest) to (\source);
    }
    \foreach \source/ \dest in {v3/v12, v5/v4, v9/v5, v12/v9, v8/v7, v7/v10, 
   								v10/v14, v14/v11}{
       \draw[IntersectionEdgeDirected] (\dest) to (\source);
    }
    \draw[IntersectionEdgeDirected] (v3) to node[left, black] {$\tilde{e}_1$} (v4);
    \draw[IntersectionEdgeDirected] (v8) to  node[right, black] {$\tilde{e}_2$} (v11);
\end{tikzpicture}
 \label{subfig:d}
 \subcaption{The graph $\tilde{T}_1$ is created from $T_1''$ by substituting each path $P_j$ by a special edge $\tilde{e}_j$. The arrows indicate an Eulerian circuit in $\tilde{T}_1$.}
\end{subfigure}
\par\bigskip
\begin{subfigure}{0.45\textwidth}
\centering
\begin{tikzpicture}[scale=0.43, auto,swap]

    \foreach \pos/\name in {{(3,8)/v1}, {(11,8)/v2}, {(1,3)/v3}, {(3,6)/v4}, 
                            {(5,5)/v5}, {(7,6)/v6}, {(9,5)/v7}, {(11,6)/v8}, 									{(5,3)/v9}, {(9,3)/v10}, {(13,3)/v11}, {(3,2)/v12},
                            {(7,2)/v13}, {(11,2)/v14}}{
        \node[vertex] (\name) at \pos {};
    }
    \foreach \source/\dest in {v1/v4, v5/v6, v6/v7, v9/v13, v2/v8}{
        \draw[T1EdgeDirected, bend left = 20] (\source) to (\dest);
        \draw[T1EdgeDirected, bend left = 20] (\dest) to (\source);
    }
    \foreach \source/ \dest in {v12/v3, v4/v5, v5/v9, v9/v12, v7/v8, v10/v7, 
   								v14/v10, v11/v14}{
       \draw[IntersectionEdgeDirectedThick, bend left = 20, line width = 1] (\source) to (\dest);
    }
	\foreach \source/ \dest in {v12/v3, v4/v5, v5/v9, v9/v12, v7/v8, v10/v7, 
   								v14/v10, v11/v14}{
       \draw[IntersectionEdgeDirectedThick, bend left = 20] (\dest) to (\source);
    }
    
\end{tikzpicture}
 \label{subfig:e}
 \subcaption{We traverse $P_j$ instead of traversing $\tilde{e}_j$. This way we obtain an Eulerian circuit $W_1''$ of $T_1''$ which contains $P_1, \dots, P_r$ as subpaths (thick edges). }
\end{subfigure} \hfill
\begin{subfigure}{0.45\textwidth}
\centering
\begin{tikzpicture}[scale=0.43, auto,swap]

    \foreach \pos/\name in {{(3,8)/v1}, {(11,8)/v2}, {(1,3)/v3}, {(3,6)/v4}, 
                            {(5,5)/v5}, {(7,6)/v6}, {(9,5)/v7}, {(11,6)/v8}, 									{(5,3)/v9}, {(9,3)/v10}, {(13,3)/v11}, {(3,2)/v12},
                            {(7,2)/v13}, {(11,2)/v14}}{
        \node[vertex] (\name) at \pos {};
    }
    \foreach \source/ \dest in {v3/v12, v4/v5, v5/v9, v9/v12, v7/v8, v7/v10, 
   								v10/v14, v14/v11}{
       \draw[IntersectionEdge] (\source) to (\dest);
    }
	\foreach \source/ \dest in {v1/v4, v1/v6, v6/v2, v2/v8, v11/v13, v13/v3}{
       \draw[edge] (\source) to (\dest);
    }
\end{tikzpicture}
 \label{subfig:f}
 \subcaption{The final tour $C_1$ is obtained by shortcutting $W_1''$ such that the subpaths $P_1, \dots, P_r$ are preserved.}
\end{subfigure}
\caption{Schematic sketch of the 4-approximation for \probl{RecovTSP}.}
\label{fig:algorithm}
\end{figure}
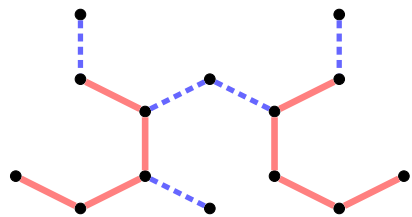
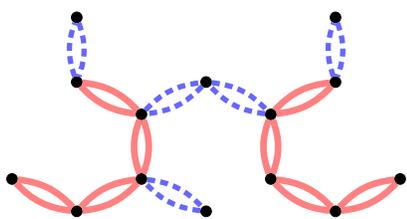
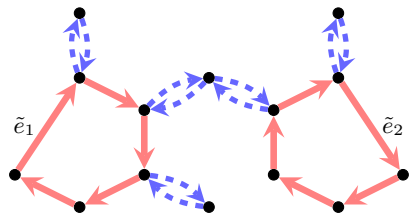
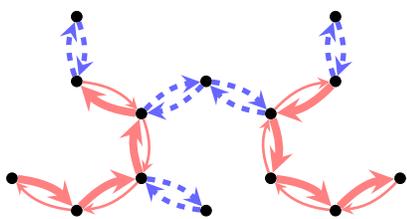
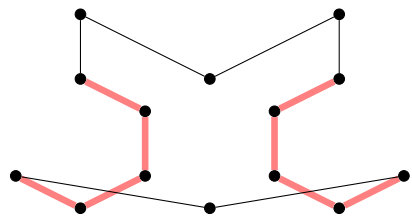

\begin{algorithm}[htb]
    \SetKwData{Left}{left}\SetKwData{This}{this}\SetKwData{Up}{up}
    \SetKwInOut{Input}{Input}\SetKwInOut{Output}{Output}
    \Input{An instance $(V,d_1, d_2, q)$ of the metric recoverable TSP, where\\
    $V$ is a set of vertices,\\
    $d_1, d_2$ are two metric distance functions on $V$,\\
    $q \in \N$ is the intersection size parameter.}
    \Output{Two tours $C_1, C_2$ on $V$ such that $|C_1 \cap C_2| \geq q$.}
    \BlankLine

    $(T_1, T_2) \gets$ Optimal solution of the \probl{RecovST} instance $(V, d_1, d_2, q)$.

    Set $T'_1 \gets T_1; T'_2 \gets T_2; \mathcal{P} \gets \emptyset$.\\
    Let $K_1, \dots, K_r$ be the connected components of $T_1 \cap T_2$.

    \ForEach{$j \in \{1,\dots,r\}$}{
       $P_j \gets$ Hamilton path in $K_j$ obtained by approximating the TSP problem on $K_j$ with respect to the distance function $d_1 + d_2$ by using the double-tree heuristic.

        Replace $K_j$ by $P_j$ in $T'_1$ and $T'_2$.

        $\mathcal{P} \gets \mathcal{P} \cup \{P_j\}$
    }

    $T''_i \gets T'_i + T'_i$ for $i=1,2$.

    $W''_i \gets $ Eulerian circuit of the graph $(V, T''_i)$ such that all paths in $\mathcal{P}$
    are simple subpaths of $W''_i$ (see \cref{lemma:euler:tour:paths}).

    Execute the shortcutting explained in \cref{shortcut:paths} on $W''_i$ to obtain 
    $C_i$ for $i=1,2$.
    
    \Return{$(C_1, C_2)$}
    \caption{Approximation algorithm for metric recoverable TSP.}
    \label{algo:tsp:approx}
\end{algorithm}\DecMargin{1em}

\subsection{Recoverable Spanning Trees and TSP}\label{sec:tress:tsp}

We are now ready to prove the correctness and the approximation guarantee of Algorithm~\ref{algo:tsp:approx}. The key observation is that the optimal 
objective value of \probl{RecovTSP} can be lower bounded by the optimum 
objective value of \probl{RecovST}. In the following lemma we assume $q < n$; note that if $q = n$, then the \probl{RecovTSP} is exactly the classical TSP and therefore we do not need to consider this case. 

\begin{lemma}\label{lemma:recovst:tsp}
    Let $(V, d_1, d_2, q)$ be an instance of the \probl{Metric RecovTSP} and let $\opt$ be its 
    optimal objective value. Let $T_1, T_2$ be an optimal solution to the 
    corresponding \probl{Metric RecovST} instance $(V, d_1, d_2, q)$. 
    If $q < n$, it holds that $d_1(T_1) + d_2(T_2) \leq \opt$.
\end{lemma}
\begin{proof}
    Let $C_1, C_2$ be any feasible solution to the \probl{RecovTSP}, i.e. they 
    are tours on $V$ such that $|E(C_1) \cap E(C_2)| \geq q$. Since $q < n$ we can select $e_1 \in E(C_1)$ and $e_2 \in E(C_2)$ 
    such that for $T'_1 := E(C_1) \setminus \{e_1\}$ and $T'_2 := E(C_2) \setminus \{e_2\}$
    it still holds that $|T'_1 \cap T'_2| \geq q$. Note that both $T'_1$ and $T'_2$ are edge sets of  
    Hamiltonian paths on $V$ and hence feasible solutions for the \probl{RecovST} instance $(V, d_1, d_2, q)$,
    implying $d_1(T_1) + d_2(T_2) \leq d_1(T'_1) + d_2(T'_2) \leq d_1(C_1) + d_2(C_2) \leq \opt$.
    \qed
\end{proof}


\subsection{Shortcutting Common Subtrees into Paths}
\label{sec:shortcut:subtrees}

We aim to obtain spanning trees that allow for a shortcutting without decreasing the size 
of the intersection. As already mentioned, we do this by substituting each of the components $K_j$ by a Hamilton path $P_j$ in $K_j$. Doing this substitution for every $j=1,\dots,r$ transforms $T_i$ into $T_i'$. The question is which paths to select. 
The following lemma shows that if the double-tree heuristic is used, we obtain an approximation guarantee of a factor 2.

\begin{lemma}\label{commonshortcut}
If for each $j =1,\dots,r$ the path $P_j$ is computed using the double tree heuristic on $K_j$ with respect to $d_1 + d_2$, then $d_1(T_1') + d_2(T_2') \leq 2(d_1(T_1) + d_2(T_2))$.
\end{lemma}
\begin{proof}
   First, observe that the tree $T_1 \cap T_2 \cap K_j$ is actually already a minimum spanning tree of $K_j$, with respect to the metric $d_1 + d_2$. 
   This fact together with the choice of $P_j$ proves that $(d_1 + d_2)(P_j) \leq 2(d_1 + d_2)(T_1 \cap T_2 \cap K_j)$. 
   By the definition of $T'_i$ we have
\begin{align*}
    d_1(T_1') + d_2(T_2')
   &= \left(\sum_{j=1}^r(d_1+d_2)(P_j)\right) + d_1(T_1 \setminus T_2) + d_2(T_2 \setminus T_1)\\
  &\leq  2(d_1(T_1) + d_2(T_2)).
\end{align*}     
   \qed
\end{proof}


Given the trees $T'_1, T'_2$ and the pairwise vertex-distinct paths $\mathcal{P}$ 
such that $\bigcup_{P \in \mathcal{P}} E(P) = T'_1 \cap T'_2$,
we let $T''_i := T'_i + T'_i$ be the multiset which contains every edge of the edge set $T'_i$ exactly twice.
The following lemma shows how to obtain an Eulerian cycle $W''_i$ in the (multi-)graph $(V, T''_i)$ 
such that each $P \in \mathcal{P}$ is a subpath of $W''_j$ for $i=1,2$.

\begin{lemma}\label{lemma:euler:tour:paths}
    For $i=1,2$, let $T''_i$ be the tree obtained in line~7 of \cref{algo:tsp:approx}.
    There exists an Eulerian tour $W''_i$ in the graph $(V, T''_i)$
    such that all paths $P \in \mathcal{P}$ are subpaths of $W''_i$.
\end{lemma}
\begin{proof}
    For both $i=1,2$, we construct a new graph $(V, \tilde{T}_i)$ by first copying the graph $(V, T''_i)$.
    Then, for each path $P_j \in \mathcal{P}$ such that $P_j = (v_1, \dots, v_{\ell})$
    we remove for each $t=1,\dots,\ell-1$ one of the two copies of the edge $\{v_t, v_{t+1}\}$ from $\tilde{T}_i$
    and add the special edge $\tilde{e}_j = \{v_{\ell}, v_1\}$ to $\tilde{T}_i$. 
    It then still holds that $(V,\tilde{T}_i)$ is Eulerian, 
    since each vertex degree stays even and the graph remains connected.
    Hence, there exists an Eulerian circuit $\tilde{W}_i$ in $(V, \tilde{T}_i)$. 
    This circuit $\tilde{W}_i$ traverses for each $P_j \in \mathcal{P}$ the previously added special edge $\tilde{e}_j$.
    We construct $W''_i$ by replacing for each $P_j \in \mathcal{P}$ the edge $\tilde{e}_j$ by the 
    path $P_j$, traversed in the corresponding direction. Note that, as claimed, $W''_i$ is 
    an Eulerian tour in $(V, T''_i)$ with each $P \in \mathcal{P}$ as a subpath.
    \qed
\end{proof}

\subsection{Shortcutting Without Skipping Paths}
\label{sec:shortcutting}

\begin{lemma}\label{shortcut:paths}
    Let $W''$ be a closed walk on $V$ and $\mathcal{P}$ be a set of pairwise vertex-disjoint subpaths of $W''$.
    Then, for any metric $d$, there exists a tour $C$ on $V$ such that 
    $d(C) \leq d(W'')$ and $C$ contains all paths in $\mathcal{P}$ as subpaths.
\end{lemma}









\begin{proof}
    We iteratively construct the tour $C$,
    by following the closed walk $W'' = (v_0, \dots, v_m)$ on $V$. Note that without 
    loss of generality $v_0 = v_m$ is not an inner vertex of any path $P \in \mathcal{P}$.
    We follow a strategy similar to the classic shortcutting applied in the double tree and Christofides' approximation algorithms.
    This means that for all vertices that are not contained in  any path in $\mathcal{P}$,
    we add them to $C$ the first time the walk $W''$ visits them. Otherwise, they are shortcut, i.e. not added to $C$.
    To ensure that no edges of 
    paths in $\mathcal{P}$ are shortcut, whenever the closed walk $W''$ visits a vertex $v$
    of any path $P \in \mathcal{P}$ and $W''$ is currently not in the first full transversal of $P$, we shortcut 
    this detour to $v$. This ensures that at the point when the closed walk 
    $W''$ traverses the path $P$ for the first time, it holds that none of the vertices in $P$ 
    are yet visited in the current subtour $C$, hence the whole path $P$ is traversed by 
    $C$. At the end of the process, we close the constructed path $C$ by adding $v_0$.

    Any vertex in $V$ appears in $C$, since we only shortcut 
    a vertex $v$ if it is already previously visited by the current subtour, 
    or if it occurs in the closed walk $W''$ before its later occurrence as part of the first transversal 
    of a path $P \in \mathcal{P}$ with $v \in P$. Also, every vertex $v\in V$ appears in $C$ exactly once,
    since $v$ also appears in $W''$ and we shortcut every time $v$ is revisited by $W''$. Hence, $C$ is a tour on $V$.
    Finally, note that any edge $\{v,w\} \in E(C)$ corresponds 
    to an edge-distinct subwalk $P_{v,w}$ of $W''$, connecting $v$ to $w$. 
    Hence, by the triangle inequality we have
    $d(C) \leq d(W'')$.
    \qed
\end{proof}

The application of \cref{shortcut:paths} in \cref{algo:tsp:approx}
transforms the closed walks $W''_i$ into tours $C_i$ for $i=1,2$ such that
it holds that
    $d_1(C_1) + d_2(C_2) \leq d_1(W''_1) + d_2(W''_2)$
    and
    $|E(C_1) \cap E(C_2)| \geq |\bigcup_{j=1}^r E(P_j)| \geq q$.
Using the results of the preceding sections we are now ready to derive 
an approximation guarantee for \cref{algo:tsp:approx}. 

\begin{proof}[Proof of \cref{thm:tsp:approx}]
By construction, we know that both $C_1, C_2$ are tours and 
$|E(C_1) \cap E(C_2)| \geq |\bigcup_{j=1}^r E(P_j)| = |T'_1 \cap T'_2| = |T_1 \cap T_2| \geq q$,
hence the tours $C_1, C_2$ are a feasible solution to \probl{RecovTSP}.
It also holds that
\begin{align*}
     d_1(C_1) + d_2(C_2) & \leq d_1(W''_1) + d_2(W''_2) \\
     &\leq 2 (d_1(T'_1) + d_2(T'_2) )
     \leq 4(d_1(T_1) + d_2(T_2))  
     \leq 4 \opt.
\end{align*}
Note that all steps of \cref{algo:tsp:approx} can be implemented in polynomial time.
\qed
\end{proof}

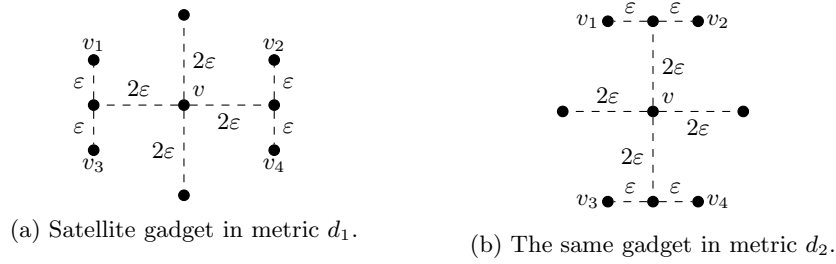
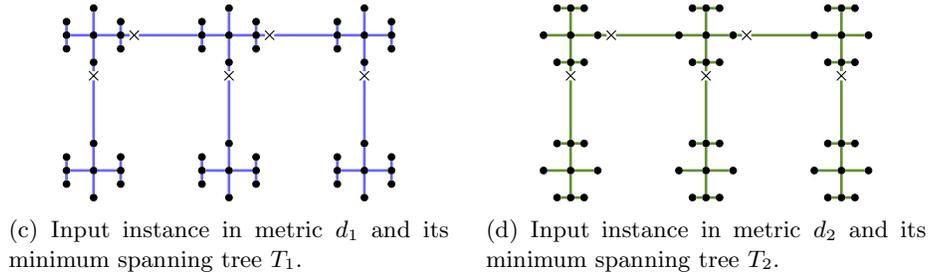
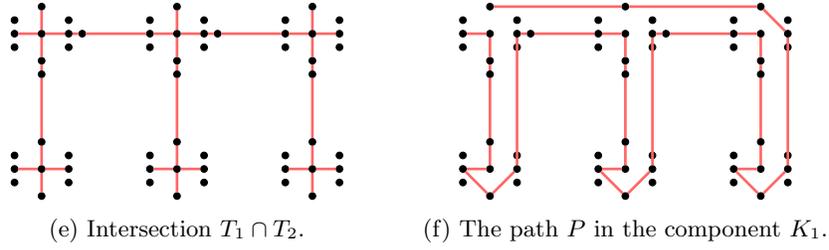
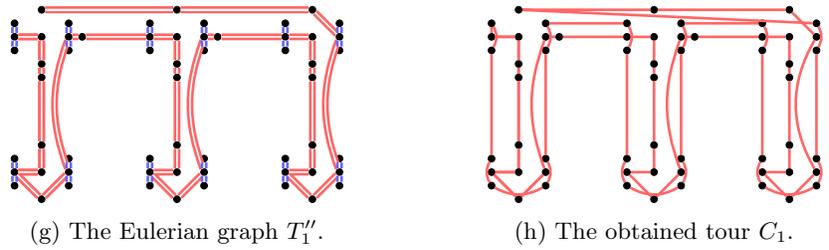
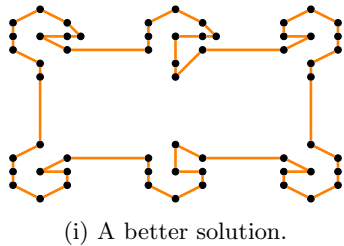
\begin{figure}[tbhp]
\begin{subfigure}{0.48\textwidth}
\centering
\begin{tikzpicture}[scale=0.6, auto,swap]
    \foreach \pos/\name in {{(0,0)/v}, {(2,0)/vE}, {(2,1)/vNE}, {(0,2)/vN}, {(-2,1)/vNW},{(-2,0)/vW}, {(-2,-1)/vSW}, {(0,-2)/vS}, {(2,-1)/vSE}}{
        \node[vertex] (\name) at \pos {};  
    }
     \node at (v) [above right] {$v$};
     \node at (vNE) [above] {$v_{2}$};
     \node at (vNW) [above] {$v_{1}$};
     \node at (vSE) [below] {$v_{4}$};
     \node at (vSW) [below] {$v_{3}$};
     \draw[edge, dashed] (v) to node {$2\varepsilon$} (vE);
     \draw[edge, dashed] (v) to node {$2\varepsilon$} (vN);
     \draw[edge, dashed] (v) to node {$2\varepsilon$} (vS);
     \draw[edge, dashed] (v) to node {$2\varepsilon$} (vW);
     \draw[edge, dashed] (vW) to node[left] {$\varepsilon$} (vNW);
     \draw[edge, dashed] (vE) to node {$\varepsilon$} (vNE);
     \draw[edge, dashed] (vE) to node[right] {$\varepsilon$} (vSE);
     \draw[edge, dashed] (vW) to node {$\varepsilon$} (vSW);
\end{tikzpicture}
 \subcaption{Satellite gadget in metric $d_1$.}
 \label{fig:tight:a}
\end{subfigure}\hfill
\begin{subfigure}{0.48\textwidth}
\centering
\begin{tikzpicture}[scale=0.6, auto,swap]
    \foreach \pos/\name in {{(0,0)/v}, {(2,0)/vE}, {(1,2)/vNE}, {(0,2)/vN}, {(-1,2)/vNW},{(-2,0)/vW}, {(-1,-2)/vSW}, {(0,-2)/vS}, {(1,-2)/vSE}}{
        \node[vertex] (\name) at \pos {};  
    }
     \node at (v) [above right] {$v$};
     \node at (vNE) [right] {$v_{2}$};
     \node at (vNW) [left] {$v_{1}$};
     \node at (vSE) [right] {$v_{4}$};
     \node at (vSW) [left] {$v_{3}$};
     \draw[edge, dashed] (v) to node {$2\varepsilon$} (vE);
     \draw[edge, dashed] (v) to node {$2\varepsilon$} (vN);
     \draw[edge, dashed] (v) to node {$2\varepsilon$} (vS);
     \draw[edge, dashed] (v) to node {$2\varepsilon$} (vW);
     \draw[edge, dashed] (vN) to node {$\varepsilon$} (vNW);
     \draw[edge, dashed] (vN) to node[above] {$\varepsilon$} (vNE);
     \draw[edge, dashed] (vS) to node[above] {$\varepsilon$} (vSE);
     \draw[edge, dashed] (vS) to node {$\varepsilon$} (vSW);
\end{tikzpicture}
 \subcaption{The same gadget in metric $d_2$.}
 \label{fig:tight:b}
\end{subfigure} \hfill
\par\bigskip
\tikzstyle{smallVertex}=[draw,circle,fill=black, minimum size=2.5pt,inner sep=0pt]
\tikzstyle{T1edge'} = [draw,blue!60,line width=1,-]
\tikzstyle{T2edge'} = [draw,OliveGreen!80!yellow,line width=1,-]
\tikzstyle{IntersectionEdge'} = [draw,red!60,line width=1,-]
\tikzset{cross/.style={cross out, draw=black, minimum size=3.5pt, inner sep=0pt, outer sep=0pt},
cross/.default={1pt}}
\begin{subfigure}{0.48\textwidth}
\centering
\begin{tikzpicture}[scale=0.18, auto,swap]
    \foreach \offsetX/\offsetY/\offsetName in {0/0/C, 2/0/E, 2/1/NE, 0/2/N, -2/1/NW, -2/0/W, -2/-1/SW, 0/-2/S, 2/-1/SE}{
    	\foreach \posX/\posY/\name in {0/0/1, 10/0/2, 20/0/3, 0/-10/4, 10/-10/5, 20/-10/6} {
    		\node[smallVertex] (\name\offsetName) at ($(\posX, \posY)+(\offsetX, \offsetY)$) {}; 
    	}         
    }
    \foreach \i/\posX/\posY in {1/0/-3, 2/3/0, 3/10/-3, 4/13/0, 5/20/-3} {
    	\node[cross] (x\i) at (\posX, \posY) {};
    }
    \draw[T1edge'] (1W) -- (1C) -- (1E) -- (x2) -- (2W) -- (2C) -- (2E) -- (x4) -- (3W) -- (3C)-- (3E);
    \draw[T1edge'] (4S) -- (4C) -- (4N) -- (x1) -- (1S) -- (1C) -- (1N);
    \draw[T1edge'] (5S) -- (5C) -- (5N) -- (x3) -- (2S) -- (2C) -- (2N);
    \draw[T1edge'] (6S) -- (6C) -- (6N) -- (x5) -- (3S) -- (3C) -- (3N);
    \draw[T1edge'] (4W) -- (4C) -- (4E);
    \draw[T1edge'] (5W) -- (5C) -- (5E);
    \draw[T1edge'] (6W) -- (6C) -- (6E);
    \foreach \i in {1,2,...,6} {
    	\draw[T1edge'] (\i NW) -- (\i W) -- (\i SW);
    	\draw[T1edge'] (\i NE) -- (\i E) -- (\i SE);
    }

\end{tikzpicture}
 \subcaption{Input instance in metric $d_1$ and its minimum spanning tree $T_1$.}
 \label{fig:tight:c}
\end{subfigure} \hfill
\begin{subfigure}{0.48\textwidth}
\centering
\begin{tikzpicture}[scale=0.18, auto,swap]
    \foreach \offsetX/\offsetY/\offsetName in {0/0/C, 2/0/E, 1/2/NE, 0/2/N, -1/2/NW, -2/0/W, -1/-2/SW, 0/-2/S, 1/-2/SE}{
    	\foreach \posX/\posY/\name in {0/0/1, 10/0/2, 20/0/3, 0/-10/4, 10/-10/5, 20/-10/6} {
    		\node[smallVertex] (\name\offsetName) at ($(\posX, \posY)+(\offsetX, \offsetY)$) {}; 
    	}         
    }
	\foreach \i/\posX/\posY in {1/0/-3, 2/3/0, 3/10/-3, 4/13/0, 5/20/-3} {
    	\node[cross] (x\i) at (\posX, \posY) {};
    }
    \draw[T2edge'] (1W) -- (1C) -- (1E) -- (x2) -- (2W) -- (2C) -- (2E) -- (x4) -- (3W) -- (3C)-- (3E);
    \draw[T2edge'] (4S) -- (4C) -- (4N) -- (x1)-- (1S) -- (1C) -- (1N);
    \draw[T2edge'] (5S) -- (5C) -- (5N) -- (x3) -- (2S) -- (2C) -- (2N);
    \draw[T2edge'] (6S) -- (6C) -- (6N) -- (x5) -- (3S) -- (3C) -- (3N);
    \draw[T2edge'] (4W) -- (4C) -- (4E);
    \draw[T2edge'] (5W) -- (5C) -- (5E);
    \draw[T2edge'] (6W) -- (6C) -- (6E);
    \foreach \i in {1,2,...,6} {
    	\draw[T2edge'] (\i NW) -- (\i N) -- (\i NE);
    	\draw[T2edge'] (\i SW) -- (\i S) -- (\i SE);
    }

\end{tikzpicture}
 \subcaption{Input instance in metric $d_2$ and its minimum spanning tree $T_2$.}
 \label{fig:tight:d}
\end{subfigure}
\par\bigskip
\begin{subfigure}{0.48\textwidth}
\centering
\begin{tikzpicture}[scale=0.18, auto,swap]
    \foreach \offsetX/\offsetY/\offsetName in {0/0/C, 2/0/E, 2/1/NE, 0/2/N, -2/1/NW, -2/0/W, -2/-1/SW, 0/-2/S, 2/-1/SE}{
    	\foreach \posX/\posY/\name in {0/0/1, 10/0/2, 20/0/3, 0/-10/4, 10/-10/5, 20/-10/6} {
    		\node[smallVertex] (\name\offsetName) at ($(\posX, \posY)+(\offsetX, \offsetY)$) {}; 
    	}         
    }
    \foreach \i/\posX/\posY in {1/0/-3, 2/3/0, 3/10/-3, 4/13/0, 5/20/-3} {
    	\node[smallVertex] (x\i) at (\posX, \posY) {};
    }
    \draw[IntersectionEdge'] (1W) -- (1C) -- (1E) -- (x2) -- (2W) -- (2C) -- (2E) -- (x4) -- (3W) -- (3C)-- (3E);
    \draw[IntersectionEdge'] (4S) -- (4C) -- (4N) -- (x1) -- (1S) -- (1C) -- (1N);
    \draw[IntersectionEdge'] (5S) -- (5C) -- (5N) -- (x3) -- (2S) -- (2C) -- (2N);
    \draw[IntersectionEdge'] (6S) -- (6C) -- (6N) -- (x5) -- (3S) -- (3C) -- (3N);
    \draw[IntersectionEdge'] (4W) -- (4C) -- (4E);
    \draw[IntersectionEdge'] (5W) -- (5C) -- (5E);
    \draw[IntersectionEdge'] (6W) -- (6C) -- (6E);

\end{tikzpicture}
 \subcaption{Intersection $T_1 \cap T_2$.}
 \label{fig:tight:e}
 \end{subfigure}
 \begin{subfigure}{0.48\textwidth}
\centering
\begin{tikzpicture}[scale=0.18, auto,swap]
    \foreach \offsetX/\offsetY/\offsetName in {0/0/C, 2/0/E, 2/1/NE, 0/2/N, -2/1/NW, -2/0/W, -2/-1/SW, 0/-2/S, 2/-1/SE}{
    	\foreach \posX/\posY/\name in {0/0/1, 10/0/2, 20/0/3, 0/-10/4, 10/-10/5, 20/-10/6} {
    		\node[smallVertex] (\name\offsetName) at ($(\posX, \posY)+(\offsetX, \offsetY)$) {}; 
    	}         
    }
    \foreach \i/\posX/\posY in {1/0/-3, 2/3/0, 3/10/-3, 4/13/0, 5/20/-3} {
    	\node[smallVertex] (x\i) at (\posX, \posY) {};
    }
    \draw[IntersectionEdge'] (1W) -- (1C) -- (1S) -- (x1) -- (4N) -- (4C) -- (4W) -- (4S) -- (4E) -- (1E) -- (x2) -- (2W) -- (2C) -- (2S) -- (x3) -- (5N) -- (5C) -- (5W) -- (5S) --(5E)-- (2E) -- (x4) -- (3W) -- (3C) -- (3S) -- (x5) -- (6N) -- (6C) -- (6W) -- (6S) --(6E) -- (3E) -- (3N) -- (2N) -- (1N);

\end{tikzpicture}
 \subcaption{The path $P$ in the component $K_1$.}
 \label{fig:tight:f}
\end{subfigure}
\par\bigskip
 \begin{subfigure}{0.48\textwidth}
\centering
\begin{tikzpicture}[scale=0.18, auto,swap]
    \foreach \offsetX/\offsetY/\offsetName in {0/0/C, 2/0/E, 2/1/NE, 0/2/N, -2/1/NW, -2/0/W, -2/-1/SW, 0/-2/S, 2/-1/SE}{
    	\foreach \posX/\posY/\name in {0/0/1, 10/0/2, 20/0/3, 0/-10/4, 10/-10/5, 20/-10/6} {
    		\node[smallVertex] (\name\offsetName) at ($(\posX, \posY)+(\offsetX, \offsetY)$) {}; 
    	}         
    }
    \foreach \i/\posX/\posY in {1/0/-3, 2/3/0, 3/10/-3, 4/13/0, 5/20/-3} {
    	\node[smallVertex] (x\i) at (\posX, \posY) {};
    }
    \foreach \offestX/\offsetY in {0/0, 0.5/0.5}
    \draw[IntersectionEdge', double] (1W) -- (1C) -- (1S) -- (x1) -- (4N) -- (4C) -- (4W) -- (4S) -- (4E);
    \draw[IntersectionEdge', double, bend left = 20] (4E) to (1E);
    \draw[IntersectionEdge', double] (1E) -- (x2) -- (2W) -- (2C) -- (2S) -- (x3) -- (5N) -- (5C) -- (5W) -- (5S) --(5E);
    \draw[IntersectionEdge', double, bend left = 20] (5E) to (2E); 
    \draw[IntersectionEdge', double] (2E) -- (x4) -- (3W) -- (3C) -- (3S) -- (x5) -- (6N) -- (6C) -- (6W) -- (6S) -- (6E);
    \draw[IntersectionEdge', double, bend left = 20] (6E)  to (3E);
    \draw[IntersectionEdge', double] (3E) -- (3N) -- (2N) -- (1N) -- (2N) -- (3N);
    \foreach \i in {1,2,...,6} {
    	\draw[T1edge', double] (\i NW) -- (\i W) -- (\i SW);
    	\draw[T1edge', double] (\i NE) -- (\i E) -- (\i SE);
    }
\end{tikzpicture}
 \subcaption{The Eulerian graph $T_1''$.}
 \label{fig:tight:g}
 \end{subfigure}\hfill
 \begin{subfigure}{0.48\textwidth}
\centering
\begin{tikzpicture}[scale=0.18, auto,swap]
    \foreach \offsetX/\offsetY/\offsetName in {0/0/C, 2/0/E, 2/1/NE, 0/2/N, -2/1/NW, -2/0/W, -2/-1/SW, 0/-2/S, 2/-1/SE}{
    	\foreach \posX/\posY/\name in {0/0/1, 10/0/2, 20/0/3, 0/-10/4, 10/-10/5, 20/-10/6} {
    		\node[smallVertex] (\name\offsetName) at ($(\posX, \posY)+(\offsetX, \offsetY)$) {}; 
    	}         
    }
    \foreach \i/\posX/\posY in {1/0/-3, 2/3/0, 3/10/-3, 4/13/0, 5/20/-3} {
    	\node[smallVertex] (x\i) at (\posX, \posY) {};
    }
    \foreach \offestX/\offsetY in {0/0, 0.5/0.5}
    \draw[IntersectionEdge'] (1W) -- (1C) -- (1S) -- (x1) -- (4N) -- (4C) -- (4W) -- (4S) -- (4E);
    \draw[IntersectionEdge', bend left = 20] (4E) to (1E);
    \draw[IntersectionEdge'] (1E) -- (x2) -- (2W) -- (2C) -- (2S) -- (x3) -- (5N) -- (5C) -- (5W) -- (5S) --(5E);
    \draw[IntersectionEdge', bend left = 20] (5E) to (2E); 
    \draw[IntersectionEdge'] (2E) -- (x4) -- (3W) -- (3C) -- (3S) -- (x5) -- (6N) -- (6C) -- (6W) -- (6S) -- (6E);
    \draw[IntersectionEdge', bend left = 30] (6E)  to (3E);
    \draw[IntersectionEdge'] (3E) -- (3N) -- (2N) -- (1N) --(3NE);
    \draw[IntersectionEdge', bend left = 30] (3NE)  to (3SE);
    \draw[IntersectionEdge'] (3SE) -- (6NE);
    \draw[IntersectionEdge', bend left = 30] (6NE)  to (6SE);
    \draw[IntersectionEdge', bend left = 30] (6SE)  to (6SW);
    \draw[IntersectionEdge', bend left = 30] (6SW)  to (6NW);
    \draw[IntersectionEdge'] (6NW) -- (3SW);
    \draw[IntersectionEdge', bend right = 30] (3SW)  to (3NW);
    \draw[IntersectionEdge'] (3NW) -- (2NE);
    \draw[IntersectionEdge', bend left = 30] (2NE)  to (2SE);
    \draw[IntersectionEdge'] (2SE) -- (5NE);
    \draw[IntersectionEdge', bend left = 30] (5NE)  to (5SE);
    \draw[IntersectionEdge', bend left = 30] (5SE)  to (5SW);
    \draw[IntersectionEdge', bend left = 30] (5SW)  to (5NW);
    \draw[IntersectionEdge'] (5NW) -- (2SW);
    \draw[IntersectionEdge', bend right = 30] (2SW)  to (2NW);
    \draw[IntersectionEdge'] (2NW) -- (1NE);
    \draw[IntersectionEdge', bend left = 30] (1NE)  to (1SE);
    \draw[IntersectionEdge'] (1SE) -- (4NE);
    \draw[IntersectionEdge', bend left = 30] (4NE)  to (4SE);
    \draw[IntersectionEdge', bend left = 30] (4SE)  to (4SW);
    \draw[IntersectionEdge', bend left = 30] (4SW)  to (4NW);
    \draw[IntersectionEdge'] (4NW) -- (1SW);
    \draw[IntersectionEdge', bend right = 30] (1SW)  to (1NW);
    \draw[IntersectionEdge'] (1NW) -- (1W);
\end{tikzpicture}
 \subcaption{The obtained tour $C_1$.}
 \label{fig:tight:h}
 \end{subfigure}
\par\bigskip
 \centering
 \begin{subfigure}{0.48\textwidth}
\centering
\begin{tikzpicture}[scale=0.18, auto,swap]
    \foreach \offsetX/\offsetY/\offsetName in {0/0/C, 2/0/E, 2/1/NE, 0/2/N, -2/1/NW, -2/0/W, -2/-1/SW, 0/-2/S, 2/-1/SE}{
    	\foreach \posX/\posY/\name in {0/0/1, 10/0/2, 20/0/3, 0/-10/4, 10/-10/5, 20/-10/6} {
    		\node[smallVertex] (\name\offsetName) at ($(\posX, \posY)+(\offsetX, \offsetY)$) {}; 
    	}         
    }
    \foreach \i/\posX/\posY in {1/0/-3, 2/3/0, 3/10/-3, 4/13/0, 5/20/-3} {
    	\node[smallVertex] (x\i) at (\posX, \posY) {};
    }
    \draw[IntersectionEdge',orange] (1SE) -- (1C) -- (1E)-- (x2) -- (1NE) -- (1N) -- (1NW) -- (1W) -- (1SW) -- (1S) -- (x1) -- (4N) -- (4NW) -- (4W) -- (4SW) -- (4S) -- (4SE) -- (4E) -- (4C) -- (4NE) -- (5NW) -- (5W) --(5SW) -- (5S) -- (5SE) -- (5E) -- (5C) -- (5N) -- (5NE) -- (6NW) -- (6C) -- (6W) --(6SW) -- (6S) -- (6SE) -- (6E) -- (6NE) -- (6N) -- (x5) -- (3S) --(3SE) -- (3E) -- (3NE) -- (3N) -- (3NW) -- (3W) -- (3C) -- (3SW) -- (2SE) -- (x3) -- (2S) -- (2C) -- (2E) -- (x4) -- (2NE) -- (2N) -- (2NW) -- (2W) -- (2SW) -- (1SE);

\end{tikzpicture}
 \subcaption{A better solution.}
 \label{fig:tight:i}
\end{subfigure}
\caption{An example instance where the 4-approximation is tight.}
\label{fig:tight}
\end{figure}

\subsection{An Example where the 4-Approximation is Tight}

We have successfully shown that Algorithm~\ref{algo:tsp:approx} provides a 4-approximation to the \probl{RecovTSP}. 
Inspired by ideas for the tightness of the double-tree algorithm~\cite{johnson1985performance}, we show:
\begin{lemma}
There exist problem instances such that Algorithm~\ref{algo:tsp:approx} can return a solution which is (asymptotically) 4 times worse than the optimal solution, even if both $d_1, d_2$ are 2-dimensional Euclidean metrics.
\end{lemma}
\begin{proof}
Let $k \geq 2$ be a fixed integer. We describe a problem instance $(V, d_1, d_2, q)$ of \probl{RecovTSP}. For each vertex $v \in V$, we assign a position $p^1_v \in \R^2$ and a possibly different position $p^2_v \in \R^2$ to it. We define $d_i(x, y) := \| p^i_x - p^i_y \|_2$ for $i=1,2$. Clearly $d_1$ and $d_2$ are Euclidean metrics. Let $0 < \varepsilon < 1/k^2$ be some small quantity.  A so-called \emph{satellite gadget} is depicted in \cref{fig:tight:a}. 
It is a gadget around some central vertex $v$ such that $v$ is surrounded by eight additional vertices. The position of these eight vertices in relation to the position of $v$ is exactly like depicted in \cref{fig:tight:a}. Here, the dashed lines between vertices symbolize their respective horizontal or vertical distance. The vertices $v_1, \dots, v_4$ are called \emph{satellites}. 
For every satellite $w$ we have $p^1_w \neq p^2_w$ and the position $p^2_w$ is like specified in \cref{fig:tight:b}. For every vertex $w$ which is not a satellite, we have $p^1_w = p^2_w$. 
Now the actual problem instance of \probl{RecovTSP} is created by considering the $2 \times k$ regular unit grid on the grid points $(i,j)_{i=1,2; j=1,\dots,k}$. \cref{fig:tight} depicts the case where $k = 3$. 
On each grid point, we place a copy of the satellite gadget. Additionally, we introduce $2k-1$ additional so called \emph{helper-vertices} at distance $3\varepsilon$ from the grid points. The instance is depicted in \cref{fig:tight:c,fig:tight:d}. Helper vertices are marked with a cross. Finally, we let $q := 12k - 1$. This completes our description of the \probl{RecovTSP} instance $(V, d_1, d_2, q)$. The following  observations can now be readily made:
\begin{itemize}
\item The unique minimum spanning tree $T_1$ ($T_2$) with respect to $d_1$ ($d_2$) is depicted in \cref{fig:tight:c} (\cref{fig:tight:d}). (The role of the helper vertices is to make the minimum spanning tree unique.)
\item The intersection $T_1 \cap T_2$ is depicted in \cref{fig:tight:e}. We have $|T_1 \cap T_2| = q$ and $(T_1, T_2)$ is the optimal solution to the \probl{RecovST}.
\item The path $P$ depicted in \cref{fig:tight:f} is a possible outcome when running line 5 of Algorithm~\ref{algo:tsp:approx}.
\item \cref{fig:tight:g} depicts the corresponding Eulerian graph $T''_1$.
\item The shortcutting procedure can run in such a way that for metric $d_1$ the tour $C_1$ depicted in \cref{fig:tight:h} is the output of Algorithm~\ref{algo:tsp:approx}. We have $d_1(C_1) = (1 + o(1))(8k - 4)$.
\item Analogously, the shortcutting procedure can run in such a way that for metric $d_2$ some tour $C_2$ is output such that $d_2(C_2) = (1 + o(1))(8k - 4)$.
\item On the other hand, consider the tour $C$ depicted in \cref{fig:tight:i}. We might have $d_1(C) \neq d_2(C)$, but still we have $d_1(C) = (1 + o(1))2k$ and $d_2(C) = (1 + o(1))2k$. If we let $C'_1 := C'_2 := C$, then $|E(C'_1) \cap E(C'_2)| = |V| \geq q$. So $(C'_1, C'_2)$ is a solution to the \probl{RecovTSP} of value $(1 + o(1))4k$. This is asymptotically a factor 4 better than $d_1(C_1) + d_2(C_2)$. This proves the lemma. \qed
\end{itemize}
\end{proof}

\subsection{Pitfalls when Applying Christofides' Algorithm}

We give two short remarks which show that some trivial ideas to modify Algorithm~\ref{algo:tsp:approx} do not work. Hence there are likely new ideas needed to obtain an approximation guarantee better than 4.

\begin{remark}\label{rem:boundtight}
    We remark that the bound provided in \cref{commonshortcut} is tight. We give an example where this is the case: Assume that the intersection  $T_1 \cap T_2$ is a star on the vertex set $U$ with $n+1$ vertices such that some vertex $u_0$ is the center of the star. Assume furthermore that for all vertices $x, y$ in $U$ we have $d_1(x, y) = d_2(x,y)$ and that the metric  used on $U$ is the Paris railway metric: Here we have $d_i(x, y) = 0$ if $x = y$, otherwise $d_i(x, y) = 1$ if $x = u_0$ or $y = u_0$, and otherwise $d_i(x, y) = 2$ for $i=1,2$. Furthermore assume that the intersection $T_1 \cap T_2$ makes up almost all of the cost of $T_1$ and $T_2$, that is $d_1(T_1) + d_2(T_2) = d_1(T_1 \cap T_2) + d_2(T_1 \cap T_2) + \varepsilon$ for some small $\varepsilon > 0$. Then $d_1(T_1) + d_2(T_2) = 2n + \varepsilon$. On the other hand, every Hamilton path in $U$ has cost at least $4n - 4$. Therefore we will have $d_1(T_1') + d_2(T_2') \geq 4n - 4$, independent of which path $P_j$ will be picked in order to replace $T_1 \cap T_2$. This example shows that even though there are better approximation algorithms known than the double-tree heuristic, using these algorithms instead of the double-tree heuristic in line 5 of Algorithm~\ref{algo:tsp:approx} does not yield a better approximation guarantee than a factor of 2 for \cref{commonshortcut}.
\end{remark}

\begin{remark}
Because the graph $(V, T''_i)$ is the double-tree of the graph $(V, T'_i)$, an approximation factor of 2 is introduced. 
One could also ask whether one can apply Christofides' algorithm to obtain some Eulerian graph $T'''_i$ from $T'_i$ plus a matching, 
and therefore only introduce a factor of $3/2$. However, this idea does not work: If one analogously transforms $T'''_i$ into  $\tilde{T_i}$, 
then one can see that even though all vertices in $\tilde{T_i}$ have even degree, one can find examples where $\tilde{T_i}$ is not connected. 
In general, one can show that there exist instances such that the cost of a minimal tour which includes all the paths $P_1,\dots, P_r$ as subpaths 
is strictly larger than the cost of $T_i$ plus a matching. This shows that Christofides' algorithm cannot trivially be applied in order to improve 
our approximation guarantee.
\end{remark}

\section{A 2-Approximation for Constant Intersection Size}\label{sec:multistage}

We now consider the setting where the required size $q$ of the intersection set is a constant number. We show that there exists a 2-approximation algorithm that can be applied to the more general \probl{$k$-St-RecovTSP}, where $k$ tours $C_1,\ldots,C_k$ with intersection size $q$ need to be constructed.

The corresponding \probl{$k$-Stage RecovST} is NP-hard according to \cite{lendl2021matroid}. Hence, it is not possible to use the same approach as in \cref{sec:recovtspapprox} to obtain a constant factor approximation algorithm for the \probl{$K$-St-RecovTSP}. 

\begin{theorem}\label{thm:tsp:approx:constantq}
    For constant $q$ and arbitrary $k$, there exists a $2$-approximation algorithm for \probl{Metric $k$-St-RecovTSP}.
\end{theorem}





    

\begin{proof}
    This result is obtained by guessing the optimal intersection 
    of the $k$ tours by checking all $\binom{n}{q}$ possibilities for subsets of pairwise vertex-disjoint paths.
    These sets of pairwise vertex-disjoint paths can then be extended to spanning trees 
    $T'_1, \dots, T'_k$ by solving $k$ instances of the minimum spanning tree problem.
    Then, the algorithm can proceed in a similar way as \cref{algo:tsp:approx},
    with the only difference that instead of performing each operation twice we now have 
    to perform them $k$ times.
    By similar arguments as in \cref{lemma:recovst:tsp} it holds that 
    $\sum_{i=1}^k d_i(T'_k) \leq \opt$. 
    Using this, we can proceed as in \cref{algo:tsp:approx} line~9 
    to obtain the tours $C_1, \dots, C_k$, for which the cost can be bounded by
        $\sum_{i=1}^k d_i(C_i) \leq \sum_{i=1}^k d_i(W''_i) 
       \leq 2 \sum_{i=1}^k d_i(T'_i) 
       \leq 2 \opt$.
    \qed
\end{proof}

\section{Implications for Recoverable Robust Optimization}
\label{sec:robust}

We now discuss the implications of our approximation results for recoverable robust optimization problems. Formally, let $\X\subseteq\{0,1\}^E$ denote the set of feasible solutions for some combinatorial optimization problem over ground set $E$, let $\cU \subseteq\mathbb{R}^E$ denote a set of cost scenarios, and let $\X^k(x) = \{ y\in\X : d(x,y) \le k\}$ denote the set of second-stage recovery solutions for some given first-stage solution $x$, where $d$ denotes some measure of distance. The recoverable robust problem is to solve
$\min_{x\in\X} \max_{c\in\cU} \min_{y\in\X^k(x)} \sum_{e\in E} C_e x_e + c_e y_e$,
see, e.g., the definition given in \cite{kasperski2016robust}.
If $\cU$ is the Cartesian product of intervals, i.e., $\cU = \{ c \in \mathbb{R}^E : c_e \in [\ell_e,u_e]\ \forall e \in E\}$, then an optimal solution to the inner maximization problem is to choose all cost coefficients to be at their upper bound $u_e$. This means that the recoverable robust problem considers only a single scenario, which is equivalent to the recoverable problem setting considered in this paper. Therefore, our approximation results hold for the \probl{Recoverable Robust TSP} with interval uncertainty.

Other uncertainty sets are considered as well, including budgeted uncertainty (see, e.g., \cite{chassein2016recoverable}, where budgeted uncertainty sets were used for the TSP). Budgeted uncertainty sets are essentially interval sets with an additional constraint on the total amount of deviation. Different variants have been proposed in the literature. In \cite{hradovich2017recoverable}, the following sets were used:
\begin{align*}
\cU^\Gamma_1 & =  \{ c \in \mathbb{R}^E : c_e = \ell_e + (u_e - \ell_e)\delta_e,\ \delta_e\in\{0,1\}\ \forall e \in E,\ \sum_{e\in E} \delta_e \le \Gamma\}, \\
\cU^\Gamma_2 & =  \{ c \in \mathbb{R}^E : c_e = \ell_e + \delta_e,\ \delta_e\in[0,u_e-\ell_e]\ \forall e \in E,\ \sum_{e\in E} \delta_e \le \Gamma\}.
\end{align*}
They showed the following result in the context of the \probl{RecovST},
 which also holds for any other combinatorial optimization problem. If $\alpha\in(0,1]$ is such that $\ell_e \ge \alpha u_e$ for all $e\in E$, then an optimal solution to the recoverable problem with respect to costs $\ell_e$ is an $1/\alpha$ approximation for the recoverable robust problem with respect to $\cU^\Gamma_1$ or $\cU^\Gamma_2$. By a straightforward modification of their proof of \cite[Lemma~6]{hradovich2017recoverable}, one can derive a $4/\alpha$-approximation algorithm for the recoverable robust TSP with budgeted uncertainty if $q$ is part of the input (and $2/\alpha$ if $q$ is constant) using our results.

\section{Conclusions}

Recoverable combinatorial optimization problems are a natural generalization of classic problems that arise in the area of robust optimization. In this paper, we considered the \probl{Recoverable Traveling Salesman Problem}, where two tours with respect to two distance functions need to be constructed, minimizing the sum of distances, such that the size of their intersection is at least a prescribed number $q$. Building upon the classic double-tree approximation idea, we showed that it is possible to transform an optimal solution 
of the \probl{RecovST}, which can be solved in polynomial time, into a feasible solution for \probl{RecovTSP} with an objective value that is at most 4 times the optimum. 
We provided an example that shows that the analysis of this algorithm is tight, and gave an intuition why it is not possible to apply Christofides' algorithm while using the same algorithmic ideas. Furthermore, we considered the case that $q$ is a constant, which allows for a stronger and easier 2-approximation algorithm, which can also be applied if more than two tours need to be constructed.

In further research, stronger approximation results are likely to exist. More specialized cases in the distance structure also seem fruitful to consider, such as the planar Euclidean case, or distance matrices with the Monge and anti-Monge property. 
Finally, it would be of interest to study the approximability of \probl{Metric Recoverable Assignment} or \probl{Matching Problems}.

\paragraph{\textbf{Funding.}} The authors acknowledge partial support by the Field of Excellence ``COLIBRI'' at the University of Graz, Deutsche Forschungsgemeinschaft (DFG) through grant GO 2069/1-1 and the Austrian Science Fund (FWF): W1230.

\bibliographystyle{splncs04} 
\bibliography{references}

\begin{thebibliography}{10}
\providecommand{\url}[1]{\texttt{#1}}
\providecommand{\urlprefix}{URL }
\providecommand{\doi}[1]{https://doi.org/#1}

\bibitem{bold2021recoverable}
Bold, M., Goerigk, M.: Recoverable robust single machine scheduling with
  interval uncertainty. arXiv preprint arXiv:2107.09310  (2021)

\bibitem{busing2011phd}
B{\"u}sing, C.: Recoverable robustness in combinatorial optimization. Cuvillier
  Verlag (2011)

\bibitem{busing2012recoverable}
B{\"u}sing, C.: Recoverable robust shortest path problems. Networks
  \textbf{59}(1),  181--189 (2012)

\bibitem{chassein2016recoverable}
Chassein, A., Goerigk, M.: On the recoverable robust traveling salesman
  problem. Optimization Letters  \textbf{10}(7),  1479--1492 (2016)

\bibitem{christofides1976worst}
Christofides, N.: Worst-case analysis of a new heuristic for the travelling
  salesman problem. Tech. rep., Carnegie-Mellon Univ Pittsburgh Pa Management
  Sciences Research Group (1976)

\bibitem{fischer2020investigation}
Fischer, D., Hartmann, T.A., Lendl, S., Woeginger, G.J.: An investigation of
  the recoverable robust assignment problem. arXiv preprint arXiv:2010.11456
  (2020)

\bibitem{hradovich2017recoverable}
Hradovich, M., Kasperski, A., Zieli{\'n}ski, P.: Recoverable robust spanning
  tree problem under interval uncertainty representations. Journal of
  Combinatorial Optimization  \textbf{34}(2),  554--573 (2017)

\bibitem{iwamasa2021optimal}
Iwamasa, Y., Takazawa, K.: Optimal matroid bases with intersection constraints:
  Valuated matroids, m-convex functions, and their applications. Mathematical
  Programming pp. 1--28 (2021)

\bibitem{johnson1985performance}
Johnson, D., Papadimitriou, C.: Performance guarantees for heuristics. In:
  Lawler, E., Lenstra, J., Kan, A.R., Shmoys, D. (eds.) The Traveling Salesman
  Problem: a guided tour of Combinatorial Optimization. Wiley, Chichester
  (1985)

\bibitem{karlin2021slightly}
Karlin, A.R., Klein, N., Gharan, S.O.: A (slightly) improved approximation
  algorithm for metric {TSP}. In: Proceedings of the 53rd Annual ACM SIGACT
  Symposium on Theory of Computing. pp. 32--45 (2021)

\bibitem{kasperski2016robust}
Kasperski, A., Zieli{\'n}ski, P.: Robust discrete optimization under discrete
  and interval uncertainty: A survey. In: Robustness analysis in decision
  aiding, optimization, and analytics, pp. 113--143. Springer (2016)

\bibitem{kasperski2017robust}
Kasperski, A., Zieli{\'n}ski, P.: Robust recoverable and two-stage selection
  problems. Discrete Applied Mathematics  \textbf{233},  52--64 (2017)

\bibitem{lachmann2021linear}
Lachmann, T., Lendl, S., Woeginger, G.J.: A linear time algorithm for the
  robust recoverable selection problem. Discrete Applied Mathematics
  \textbf{303},  94--107 (2021)

\bibitem{lendl2021matroid}
Lendl, S., Peis, B., Timmermans, V.: Matroid bases with cardinality constraints
  on the intersection. Mathematical Programming pp. 1--24 (2021)

\bibitem{liebchen2009concept}
Liebchen, C., L{\"u}bbecke, M., M{\"o}hring, R., Stiller, S.: The concept of
  recoverable robustness, linear programming recovery, and railway
  applications. In: Robust and online large-scale optimization, pp. 1--27.
  Springer (2009)

\bibitem{rosenkrantz1977analysis}
Rosenkrantz, D.J., Stearns, R.E., Lewis, II, P.M.: An analysis of several
  heuristics for the traveling salesman problem. SIAM journal on computing
  \textbf{6}(3),  563--581 (1977)

\end{thebibliography}

\end{document}